\documentclass[11pt]{article}%
\usepackage[margin=1in]{geometry}
\usepackage{amssymb,amsmath}
\usepackage{graphicx, epsfig}
\usepackage{float}

\usepackage{verbatim}
\usepackage{wrapfig}
\usepackage{stmaryrd}

\newcommand{\IR}{\mathbb{R}}
\newcommand{\eps}{\varepsilon}
\newcommand{\IntRange}[1]{\left\llbracket #1 \right\rrbracket}
\def\vgap{\vspace{0.1 in}}

\newcommand{\cP}{\mathcal{P}}
\newcommand{\cC}{\mathcal{C}}
\newcommand{\cT}{\mathcal{T}}
\newcommand{\oG}{\overline{G}}

\newcommand{\cG}{\mathcal{G}}

\newcommand{\Patrascu}{P\v{a}tra\c{s}cu}

\newtheorem{theorem}{Theorem}
\newtheorem{lemma}{Lemma}

\newenvironment{proof}{\trivlist\item[]\emph{Proof}:}%
{\unskip\nobreak\hskip 1em plus 1fil\nobreak$\Box$
\parfillskip=0pt%
\endtrivlist}

\begin{document}

\title{Orthogonal Point Location and  Rectangle Stabbing Queries in 3-d}
\author{Timothy M. Chan\thanks{Dept.\ of Computer Science, University of Illinois at Urbana-Champaign, USA. Email 
{\tt tmc@illinois.edu}.}
\and Yakov Nekrich\thanks{Cheriton School of Computer Science, University of Waterloo, Canada. Email {\tt yakov.nekrich@googlemail.com}.}
\and Saladi Rahul\thanks{Dept.\ of Computer Science, University of Illinois at Urbana-Champaign, USA. Email
{\tt saladi@uiuc.edu}.}
\and Konstantinos Tsakalidis\thanks{Dept.\ of Computer and Information Science, 
Tandon School of Engineering, New York University, USA. Email {\tt kt79@nyu.edu}
Partially supported by NSF grants CCF-1319648 and CCF-1533564.}
}
\date{}
\maketitle




\begin{abstract}
 In this work, we present a collection of new results on
 two fundamental problems in geometric data structures: 
{\em orthogonal point location} and {\em rectangle stabbing}.

\begin{itemize}
\item {\bf Orthogonal point location.}
We give the first linear-space data structure that supports 3-d point location queries on $n$ disjoint axis-aligned boxes 
with \emph{optimal\/} $O\left( \log n\right)$ query time
in the (arithmetic) pointer machine model.  This improves the previous $O\left( \log^{3/2} n \right)$ bound of Rahul [SODA 2015].  
We similarly obtain the first linear-space data structure in the I/O
model with optimal query cost, and also the first linear-space data structure in the 
word RAM model with sub-logarithmic query time.

\item {\bf Rectangle stabbing.}
We give the first linear-space data structure that supports 3-d
$4$-sided and $5$-sided rectangle stabbing queries 
in \emph{optimal\/} $O(\log_wn+k)$ time in the word RAM model.
We similarly obtain the first optimal data structure for the closely related problem of 2-d top-$k$ rectangle stabbing in the word RAM model, and also improved results
for 3-d 6-sided rectangle stabbing.
\end{itemize}

For point location, our solution is simpler than previous methods, and is based on an interesting variant of the van Emde Boas recursion, applied in a round-robin fashion over the dimensions, combined with bit-packing techniques.
For rectangle stabbing, our solution is a variant
of Alstrup, Brodal, and Rauhe's grid-based recursive technique (FOCS 2000), combined with
a number of new ideas.
 \end{abstract}

\thispagestyle{empty}
\newpage

\section{Introduction}

In this work we present a plethora of new results on two fundamental problems in geometric data structures: (a) {\em orthogonal point location} (where the input rectangle or boxes are 
non-overlapping), and (b) {\em rectangle stabbing} (where the input rectangles or boxes are 
overlapping).

\subsection{Orthogonal point location}
{\em Point location} is among the most central problems in the field of computational geometry, which is covered in textbooks and has countless applications. In this paper we study the {\em orthogonal point location} problem. 
Formally, we want to preprocess a set of $n$ {\em disjoint} axis-aligned  boxes (hyperrectangles) in $\IR^d$ into a data structure, so that the box in the set containing a given query point (if any) can be reported efficiently. There are two natural versions of this problem, for (a) {\em arbitrary disjoint boxes} where the input boxes need not fill the entire space, and 
(b) a {\em subdivision} where the input boxes fill the entire space. 

\subparagraph*{Arbitrary disjoint boxes.} Historically, the point location problem has been studied in the pointer machine model and the main question  has been the following: 
\begin{center}
{\em ``Is there a linear-space structure with $O(\log n)$ query time?''}
\end{center}
In 2-d this question has been successfully resolved: 
there exists a linear-space structure with $O(\log n)$ query time~\cite{lt80,k83,egs86,st86,snoeyink} 
(actually this result holds for nonorthogonal point location). In 3-d there has been work on this problem~\cite{ehh86,il00,aal10,r15}, 
but the question has not yet been resolved. 
The currently best known result on the pointer machine model is a linear-space structure with  $O(\log^{3/2} n)$ query time by Rahul~\cite{r15}. In this paper, 

\begin{itemize}
\item we obtain the  first linear-space structure with $O(\log n)$ query time 
for 3-d orthogonal point location for arbitrary disjoint boxes. The structure works in the (arithmetic) pointer machine model and is {\em optimal} in this model. 
\end{itemize}

The orthogonal point location problem has  been studied in the I/O-model and the word RAM as well 
(please see Section A in the appendix for a brief description of these models).
In the I/O model, an optimal solution is known in 2-d~\cite{gtvv93,adt03}:
a linear-space structure with $O(\log_Bn)$ query time, 
where $B$ is the block size
(this result holds for nonorthogonal point location). 
However, in 3-d  the best known result is a linear-space structure with a query cost of $O(\log_B^2n)$ I/Os by Nekrich~\cite{n08} (for orthogonal
point location for disjoint boxes). 

\begin{itemize}
\item
In the I/O model,
we obtain the first  linear-space  structure with $O(\log_Bn)$ query cost 
for 3-d orthogonal point location for arbitrary disjoint boxes. 
This result is {\em optimal}. 
\end{itemize}

In the word RAM model, an optimal solution in 2-d was given by
Chan~\cite{c13} with a query time of $O\left(\log \log U\right)$, assuming that  input coordinates are in 
 $\left[ U \right]=\{0,1,\ldots,U-1\}$. However, in 3-d the best known result for arbitrary disjoint boxes is a linear-space structure with  
$O\left( \log n\log\log n \right)$ query time: this result was not stated explicitly before but can obtained by 
an interval tree augmented with Chan's 2-d orthogonal point location structure~\cite{c13} at each node.  Our above new result with logarithmic query time
is already an improvement even in the word RAM, but we can do slightly better still:

\begin{itemize}
\item In the $w$-bit word RAM model, we obtain the first linear-space structure with  
\emph{sub-logarithmic\/} query time for 3-d orthogonal point location for arbitrary disjoint boxes. 
The time bound is $O(\log_wn)$.
(We do not know whether this 
result is optimal, however.) 
\end{itemize}

\subparagraph*{Subdivisions.}
In the plane, the two versions of the problem are equivalent in the sense that any arbitrary set of $n$ disjoint rectangles can be converted into a subdivision of $\Theta\left( n\right)$ rectangles via the vertical decomposition. 
In 3-d, the two versions are no longer equivalent, 

\begin{wrapfigure}{r}{0.2\textwidth}
\centering
\includegraphics[scale=0.55]{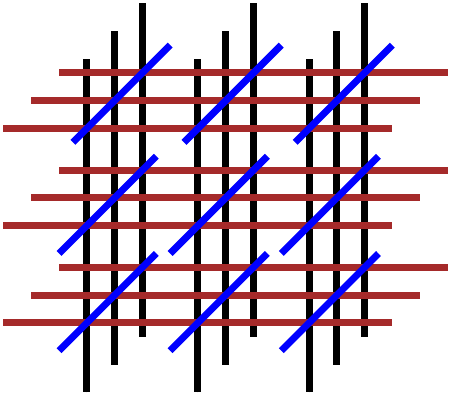}
\end{wrapfigure}
\noindent
since there exist sets of $n$ disjoint boxes that need $\Omega\left( n^{3/2}\right)$ boxes to fill the entire space. See 
figure on the right.

In 3-d  the special case of a subdivision is potentially easier than the arbitrary disjoint boxes setting, 
as the former allows for a fast $O(\log^2\log U)$ query time in the word RAM model with $O(n\log\log U)$ space,
as shown by de Berg, van Kreveld, and Snoeyink~\cite{bks95} (with an improvement by Chan~\cite{c13}).

\begin{itemize}
\item In the word RAM model, we further improve de Berg, van Kreveld, and Snoeyink's method
to achieve a {\em linear}-space structure with 
$O(\log^2\log U)$ query time for 3-d orthogonal point location on subdivisions.
\end{itemize}

\subsection{Rectangle stabbing}
Rectangle stabbing is a classical problem in geometric data structures~\cite{p08b,aal12,b77b,c86, r15}, which is as old, and as equally natural, as orthogonal range searching---in fact, it can be viewed as an ``inverse'' of 
orthogonal range searching, where the input objects are
boxes and query objects are points, instead of vice versa.
Formally, we want to preprocess a set $S$ 
 of $n$ axis-aligned boxes (possibly overlapping) in $\IR^d$ into a data structure, 
so that the boxes in $S$ containing a given query point $q$ can be reported efficiently. 
(As one of many possible applications, imagine a dating website, where each lady
is interested in gentlemen whose salary is 
in a range $[S_1{,}S_2]$ and age is in a range $[A_1{,}A_2]$; suppose that a gentleman with salary $x_q$ and age $y_q$ wants to identify all ladies who might be potentially interested in him.)
 
Throughout this paper, we will assume that the endpoints of the 
rectangles lie on the grid $[2n]^3$ (this can be achieved via a simple 
rank-space reduction).
In the word RAM model, \Patrascu~\cite{p11} gave a lower bound of $\Omega(\log_wn)$ query time for any data 
structure which occupies at most $n\log^{O(1)}n$ space to answer the 2-d rectangle stabbing query. 
Shi and Jaja~\cite{sj05b} presented an optimal solution in 2-d which occupies 
 linear space  with $O(\log_wn + k)$ query time,
where $k$ is the number of rectangles reported. 

\begin{wrapfigure}{r}{0.5\textwidth}
\vspace{-2ex}
 \centering
\includegraphics[scale=0.7]{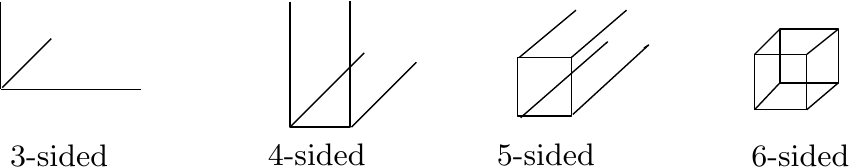}
\label{fig:rectangles}
\end{wrapfigure}
We  introduce some notation to define various types of rectangles in 3-d.
 (We will use the terms ``rectangle'' and ``box'' interchangably throughout the paper.)
A rectangle in 3-d is called \emph{$(3+t)$-sided\/} if it is bounded in 
$t$ out of the $3$ dimensions and unbounded (on one side) in the remaining $3-t$ dimensions.


In the word RAM model, an optimal solution in 3-d is known {\em only} for the
$3$-sided rectangle stabbing query: a linear-space structure with 
$O(\log\log_wn+k)$ query time (by combining the work of Afshani~\cite{p08b} and 
Chan~\cite{c13}; this is optimal due to the lower bound of \Patrascu~and~Thorup~\cite{pt06}). 
Finding an optimal solution 
for $4$-, $5$-, and $6$-sided rectangle stabbing has remained open.

\subparagraph*{3-d $4$- and $5$-sided rectangle stabbing.}
Currently, the best-known result  for $4$-sided and $5$-sided rectangle stabbing queries
by Rahul~\cite{r15} occupies $O(n\log^{*}n)$ 
space with $O(\log n +k)$ and $O(\log n\log\log n +k)$ query time, respectively.
This result holds in the pointer machine model. 
For $4$-sided rectangle stabbing, adapting Rahul's solution to the word RAM model 
does not lead to any improvement in  the query time (the bottleneck is in 
answering $\log n$ 3-d dominance reporting queries). 
 For $5$-sided rectangle stabbing, even if we assume the existence of an 
 optimal $4$-sided rectangle stabbing 
structure, plugging it into  Rahul's solution can 
improve the query time to only $O(\log n + k)$,
which is still suboptimal.
In this paper,

\begin{itemize}
\item we obtain the first {\em optimal} solution for 3-d $4$-sided and $5$-sided rectangle stabbing 
in the word RAM model: a linear-space structure with $O(\log_wn+k)$ query time.
\end{itemize}

\subparagraph{2-d top-$k$ rectangle stabbing.}
Recently, there has been a lot of interest in 
top-$k$ range searching
~\cite{abz11,b16,bfgl09,rj14,rt15,rt16,st12,t14}.
Specifically, in the 2-d top-$k$ rectangle stabbing problem, we want to preprocess a 
set  of {\em weighted} axis-aligned rectangles in 2-d, so 
that given a query point $q$ and an integer $k$, the goal is to report the $k$ largest-weight
rectangles containing (or stabbed by) $q$.
This problem is closely related to the $5$-sided rectangle stabbing problem
(by treating the weight as a third dimension, a rectangle $r$ with weight $w(r)$ 
can be mapped to a $5$-sided rectangle $r\times (-\infty,w(r)]$).

\begin{itemize}
 \item By extending the solution for 3-d $5$-sided rectangle stabbing problem, we obtain the first {\em optimal} solution 
for the 2-d top-$k$ rectangle stabbing problem: a linear-space structure with $O(\log_wn+k)$ query time. 
\end{itemize}

\subparagraph*{3-d $6$-sided rectangle stabbing.}
Our new solution to 3-d $5$-sided rectangle stabbing, combined with
standard interval trees,
immediately implies a solution to 3-d $6$-sided rectangle stabbing with a
query time of $O(\log_wn\cdot\log n +k)$, which is already new. 
But we can do slightly better still:

\begin{itemize}
\item We obtain a linear-space structure with $O(\log_w^2n+k)$ query time 
for 3-d $6$-sided rectangle stabbing problem in the word RAM model.
We conjecture this to be optimal (the analogy is the lower bound of $\Omega(\log^2n+k)$ 
query time for linear-space pointer machine structures~\cite{aal12}).
\end{itemize}

\subparagraph*{Back to orthogonal point location.}  Our solution for orthogonal point location uses rectangle stabbing as a subroutine: if there is  an $S(n)$-space data structure with $Q(n)+O(k)$ query time 
to answer the rectangle stabbing problem in $\Re^d$, then one can  obtain a data structure for 
orthogonal point location in $\Re^{d+1}$ with $O(S(n))$-space and $O(Q(n))$ time.
By plugging in our new results for 3-d $6$-sided rectangle stabbing, 
we obtain a linear-space word RAM 
structure which can answer any orthogonal point location query in 4-d in $O(\log_w^2n)$  time, improving the previously known $O(\log^{2} n\log\log n)$ bound~\cite{c13}.

\subsection{Our techniques}\label{intro:tech}

Our results are obtained using a number of new ideas (in
addition to existing data structuring techniques), which we feel
are as interesting as the results themselves.

\subparagraph*{3-d orthogonal point location.} 
To better appreciate our new 3-d orthogonal point location method,
we first recall that the current best word-RAM method had
$O(\log n\log\log n)$ query time, and was obtained by
building an interval tree over the $x$-coordinates, and
at each node of the tree, storing Chan's 2-d point location
data structure on the $yz$-projection of the rectangles.
Interval trees caused the query time to increase by a logarithmic factor,
while Chan's 2-d structures achieved $O(\log\log n)$ query time
via a complicated van-Emde-Boas-like recursion.
We can thus summarize this approach loosely by
the following recurrence for the query time (superscripts refer to the
dimension):
\[Q^{(3)}(n) = O(Q^{(2)}(n)\log n) \text{ and } Q^{(2)}(n) =Q^{(2)}(\sqrt{n}) + O(1)\implies Q^{(3)}(n)=O(\log n\log\log n).
\]
(Note that naively increasing the fan-out of the interval tree could reduce the query time
but would blow up the space usage.)

In the pointer machine model, 
the current best data structure
by Rahul~\cite{r15}, with $O(\log^{3/2}n)$ query time, required an even more complicated combination of interval trees, 
Clarkson and Shor's random sampling technique, 
3-d rectangle stabbing, 
and 2-d orthogonal point location.

To avoid the extra $\log\log n$ factor, we cannot afford to
use Chan's 2-d orthogonal point location structure as a subroutine;
and we cannot work with just $yz$-projections, which
intuitively cause loss of efficiency.
Instead, we propose a more direct solution based on a
new van-Emde-Boas-like recursion, aiming for a new recurrence of the form
\[Q^{(3)}(n) =Q^{(3)}(\sqrt{n}) + O(\log n).\]
The $O(\log n)$ term arises from the need to solve 2-d rectangle
stabbing subproblems, on projections along all three
directions (the $yz$-, $xz$-, and $xy$-plane), applied in a 
\emph{round-robin} fashion.
The new recurrence then solves to $O(\log n)$---notice how $\log\log$
disappears, unlike the usual van Emde Boas recursion!
In the word RAM model, we can even use known sub-logarithmic
solutions to 2-d rectangle stabbing to get $O(\log_w n)$ query time.

We emphasize that our new method is much \emph{simpler} than the previous,
slower methods, and is essentially self-contained except for  the use of a
known data structure for 2-d rectangle stabbing emptiness (which reduces to
standard 2-d orthogonal range counting).

One remaining issue is space.
In our new method, a rectangle is stored $O(\log\log n)$ times, due
to the depth of the recursion.  To achieve linear space, we
need another idea, namely, \emph{bit-packing} tricks, to compress
the data structure.  Because of
the rapid reduction of the universe size in the round-robin van-Emde-Boas
recursion, the amortized space in words per input box satisfies a recurrence
of the form
\[ s(n) =s(\sqrt{n}) + O\left(\frac{\log n}{w}\right)
\implies s(n)= O\left(\frac{\log n}{w}\right)=O(1).\]
Our new result on the subdivision case is obtained by a similar space-reduction trick.

\subparagraph*{3-d $5$-sided rectangle stabbing.} 
For 3-d rectangle stabbing, the previous solution by
Rahul~\cite{r15} was based on a grid-based, $\sqrt{n}$-way recursive approach
of Alstrup, Brodal, and Rauhe~\cite{abr00}, originally designed
for 2-d orthogonal range searching.  The fact that the approach
can be adapted here is nontrivial and interesting, since our input
objects are now more complicated (rectangles instead of points) and
the target query time is quite different (near logarithmic rather than
$\log\log$).  More specifically, Rahul first solved
the 4-sided case via a complicated data structure, 
and then applied Alstrup et al.'s technique
to reduce $5$-sided rectangles to $4$-sided rectangles, 
which led to a query-time recurrence similar to the following (subscripts denote
the number of sides, and output cost related to $k$ is ignored):
\[ Q_4(n) = O(\log n) \text{ and }
  Q_5(n) = 2Q_5(\sqrt{n}) + O(Q_4(n))
  \implies Q_5(n)=O(\log n\log\log n). \]

Intuitively, the reduction from the 5-sided to the 4-sided case
causes loss of efficiency.  To avoid the extra $\log\log n$ factor,
we propose a new method that is also based on Alstrup {\em et al.}'s recursive technique, but reduces 5-sided rectangles 
directly to $3$-sided rectangles, aiming for a
new recurrence of the form
\[ Q_3(n)=O(\log\log_w n) \text{ and }
Q_5(n) = 2Q_5(\sqrt{n})+ O(Q_3(n)).\]
During recursion, we do not put 4-sided rectangles in separate structures
(which would slow down querying), but instead use a common tree
for both 4-sided and 5-sided rectangles.
The new recurrence then
solves to $Q_5(n)=O(\log_wn)$ with an appropriate base case---notice
how $\log\log$ again disappears, and notice how this gives a new result
even for the 4-sided case!

One remaining issue is space.  Again, we can compress the data structure by
incorporating bit-packing tricks
(which was also used in Alstrup {\em et al.}'s original method).
For 4- and 5-sided rectangle stabbing, the space recurrence then solves
to linear. 

However, with space compression, a new issue arises.  The cost of reporting
each output rectangle in a query increases to $O(\log\log n)$ (the
depth of the recursion), because
of the need to \emph{decode} the coordinates of a compressed rectangle.
In other words, the query cost becomes $O(\log_w n + k\log\log n)$
instead of $O(\log_w n + k)$.
This extra decoding overhead also occurred in previous work on 2-d orthogonal
range searching by Alstrup {\em et al.}~\cite{abr00} and
Chan {\em et al.}~\cite{clp11}, and it is open how to avoid the overhead
for that problem without sacrificing space  (this is related to the so-called
\emph{ball inheritance problem}~\cite{clp11}).

We observe that for the 4- and 5-sided rectangle stabbing problem, a 
surprisingly simple idea suffices
to avoid the overhead: instead of keeping
pointers between consecutive levels of the recursion tree, we just keep
pointers directly from each level to the \emph{leaf} level.



\subparagraph*{3-d $6$-sided rectangle stabbing.} 
We can solve $6$-sided rectangle stabbing by using
our result for $5$-sided rectangle stabbing as a subroutine. 
However, the naive reduction via interval trees increases the query time by a $\log n$ factor
instead of $\log_w n$.  To speed up querying,
the standard idea is to use a tree with a larger fan-out $w^\varepsilon$.
This leads to various
colored generalizations of
2-d rectangle stabbing with a small number $w^{\varepsilon}$ of colors.
Much of our ideas can be extended to solve these colored subproblems in a straightforward
way, but a key subproblem, of answering colored 2-d dominance searching queries
in $O(\log\log_w n + k)$ time with linear space, is nontrivial.  We solve
this key subproblem via a clever use of 2-d shallow cuttings,
combined with a grouping trick, which may be of independent interest.

 

\section{Orthogonal Point Location in 3-d}\label{sec:opl3d}

\subparagraph{Preliminaries.}
Our solution to 3-d orthogonal point location will require
known data structures for \emph{2-d orthogonal point location} and \emph{2-d rectangle stabbing emptiness}.

\begin{lemma}\label{lem:opl2d}
Given $n$ disjoint axis-aligned rectangles in $[U]^2\ (n\le U\le 2^w)$,
there are data structures for point location with
$O\left(\frac{n\log U}{w}\right)$ words of space and
\begin{quote}
\begin{itemize}
\item
$O\left(\log n\right)$ query time in the pointer machine model;
\item
$O\left(\log_B n\right)$ query cost in the I/O model;
\item
$O\left(\min\{\log\log U, \,\log_w n\}\right)$ query time in the word RAM model.
\end{itemize}
\end{quote}
\end{lemma}


\begin{lemma}\label{lem:stab2d}
Given $n$ (possibly overlapping) axis-aligned rectangles in $[U]^2\ (n\le U\le 2^w)$,
there are data structures for rectangle stabbing emptiness with
$O\left(\frac{n\log U}{w}\right)$ words of space and
\begin{quote}
\begin{itemize}
\item
$O\left(\log n\right)$ query time in the pointer machine model;
\item
$O\left(\log_B n\right)$ query cost in the I/O model;
\item
$O\left(\log_w n\right)$ query time in the word RAM model.
\end{itemize}
\end{quote}
\end{lemma}
Proofs of Lemmata~\ref{lem:opl2d} and~\ref{lem:stab2d} are presented in Appendix~\ref{app:missing-details}. 


\subparagraph{Data structure.} We are now ready to describe our data structure for 3-d orthogonal point location. We focus on the pointer machine model first. 
At the beginning, we apply a rank space reduction (replacing input coordinates by their ranks) so that all coordinates are
in $[2n]^3$, where $n$ is the global number of input boxes.  Given a query point, we can initially find the ranks of its coordinates by three predecessor searches (costing $O(\log n)$ time in the pointer machine model).

We describe our preprocessing algorithm recursively.
The input to the preprocessing algorithm is a set of $n$ disjoint boxes that are assumed to be aligned to the $\left[ U_x\right] \times \left[ U_y\right] \times \left[ U_z\right]$ grid.  (At the beginning, $U_x=U_y=U_z=2n$.)

Without loss of generality, assume that $U_x\ge U_y,U_z$.
We partition the $\left[ U_x\right] \times \left[ U_y\right] \times \left[ U_z\right]$ grid into $\sqrt{U_x}$ equal-sized vertical slabs perpendicular to the $x$-direction. See Figure~\ref{fig:structure}.
(In the symmetric case $U_y\ge U_x,U_z$ or $U_z\ge U_x,U_y$, we partition along the $y$- or $z$-direction instead.)  We classify the boxes into two categories:

\begin{itemize}
	\item {\em Short boxes.} For each slab, define its short boxes to be those that lie completely inside the slab.
	
	\item {\em Long boxes.} Long boxes intersect the boundary (vertical plane) of at least one slab. Each long box $\mathcal{B}$ is broken into three disjoint boxes:
	\begin{itemize}
		\item {\em Left box.} Let $s_L$ be the slab containing the left endpoint (with respect to the $x$-axis) of $\mathcal{B}$. The left box is defined as $\mathcal{B} \cap s_L$.
		
		\item {\em Right box.} Let $s_R$ be the slab containing the right endpoint of $\mathcal{B}$. The right box is defined as $\mathcal{B} \cap s_R$.
		
		\item {\em Middle box.} The remaining portion of box $\mathcal{B}$ after removing its left and right box, i.e. $\mathcal{B} \setminus \left( \left( \mathcal{B}\cap s_L\right)\cup \left( \mathcal{B}\cap s_R\right)\right)$.
	\end{itemize}
\end{itemize}
\begin{figure}[h]
	\centering
	\includegraphics[scale=0.8]{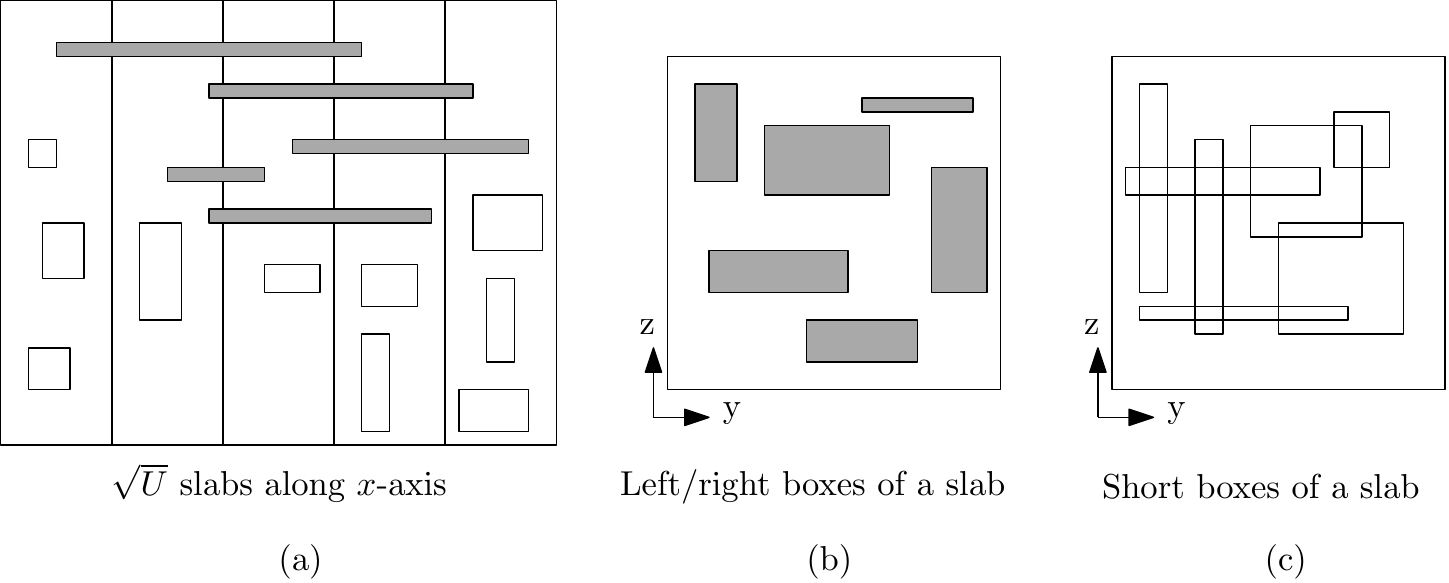}
	\caption{Boxes obtained after partitioning along the $x$-direction.}
	\label{fig:structure}
\end{figure}

We build our data structure as follows: 

\begin{enumerate}
	\item {\em Planar point location structure.} For each slab, we project its left boxes onto the $yz$-plane. The projected boxes remain disjoint, since they intersect a common boundary. We store them in a data structure  for 2-d orthogonal point location by Lemma~\ref{lem:opl2d}. We do this for the slab's right boxes as well.
	
	\item {\em Rectangle stabbing structure.} For each slab, we project its short boxes onto the $yz$-plane. The short boxes are not necessarily disjoint. We store them in a data structure for 2-d
rectangle stabbing emptiness by Lemma~\ref{lem:stab2d}.
	
	\item {\em Recursive middle structure}. We recursively build a {\em middle structure} on all the middle boxes.
	
	\item {\em Recursive short  structures}. For each slab, we recursively 
	build a {\em short structure} on all the short boxes inside the slab.
\end{enumerate}

By translation or scaling, these recursive short structures
or middle structure
can be made aligned to
the $\left[ \sqrt{U_x}\right] \times \left[ U_y \right]  \times \left[ U_z \right]$ grid.   
In addition, we store the mapping from left/right/middle boxes to their original boxes, as a list of pairs (sorted lexicographically) packed in $O\left(\frac{n \log (U_xU_yU_z)}{w}\right)$ words.


\subparagraph{Query algorithm.} The following lemma is crucial for deciding whether to query recursively the middle or the short structure.

\begin{lemma}
	Given a query point $\left( q_x, q_y , q_z \right)$, if the query with $(q_y,q_z)$ on the rectangle stabbing emptiness structure of the slab that contains $q_x$ returns
	\begin{itemize}
		\item {\sc Non-empty}, then the query point cannot lie inside a box stored in the middle structure, or 
		\item {\sc Empty}, then the query point cannot lie inside a box stored in the slab's short structure.
	\end{itemize}
\end{lemma}

\begin{proof}
If {\sc Non-empty} is returned, then the query point is stabbed by the extension (along the $x$-direction) of a box in the slab's short structure and cannot be stabbed by
any box stored in the middle structure, because of disjointness of the input boxes. If {\sc Empty} is returned, then obviously the query point cannot lie inside a box stored 
in the short structure.
\end{proof}

To answer a query for a given point $\left( q_x, q_y , q_z \right)$, we proceed as follows:

\begin{enumerate}
	\item Find the slab that contains $q_x$ by predecessor search over the slab boundaries. 
	\item Query with $\left( q_y , q_z \right)$ the planar point location structures at this slab. If a left or a right box returned by the query contains the query point, then we are done. 
	\item 
	Query with $\left( q_y , q_z \right)$ the rectangle stabbing emptiness structure at this slab. If it returns {\sc Non-empty}, query recursively the slab's short structure, else query recursively the middle structure (after appropriate translation/scaling of the query point).   
\end{enumerate}

In step~3, to decode the coordinates of the output box, we need to map from a left/right/middle box to its original box; this can be done naively by another predecessor search in the list of pairs we have stored.

\subparagraph{Query time analysis.}
Let $Q\left( U_x, U_y, U_z\right)$ denote the query time for our data structure in the  $\left[ U_x\right]  \times \left[ U_y\right] \times \left[ U_z\right]$ grid.  Observe that the number of boxes $n$ is trivially upper-bounded by $U_xU_yU_z$ because of disjointness.
The predecessor search in step~1,
the 2-d point location query in step~2, and
the 2-d rectangle stabbing query in step~3 all take
$O\left(\log n\right)=O\left(\log(U_xU_yU_z)\right)$ time
by Lemmata~\ref{lem:opl2d} and~\ref{lem:stab2d}.  
We thus obtain the following recurrence, assuming that $U_x\ge U_y,U_z$:
\[
Q\left( U_x, U_y, U_z\right) = Q\left( \sqrt{U_x}, U_y, U_z\right) + O\left( \log \left(U_x U_y U_z\right) \right).
\]
If $U_x=U_y=U_z=U$, then three rounds of recursion will partition along the $x$-, $y$-, and $z$-directions and decrease $U_x$, $U_y$, and $U_z$ in a round-robin fashion, yielding
\[Q\left( U, U, U \right) = Q\left( \sqrt{U},\sqrt{U},\sqrt{U} \right) + O\left(\log U\right),
\]
which solves to $Q\left(U,U,U\right)=O\left(\log U\right)$.
As $U=2n$ initially, we get $O(\log n)$ query time.

\subparagraph{Space analysis.}
Let $s\left( U_x,U_y,U_z \right)$ denote the \emph{amortized} number of words of space needed
per input box for our data structure in the $\left[ U_x\right] \times  \left[ U_y \right]  \times \left[ U_z \right]$ grid. 
The amortized number of words per input box for the 2-d point location and
rectangle stabbing structures is 
$O\left( \frac{\log(U_xU_yU_z)}{w}\right)$
by Lemmata~\ref{lem:opl2d} and~\ref{lem:stab2d}.  
We thus obtain the following recurrence, assuming that $U_x\ge U_y,U_z$:
\[
s\left( U_x, U_y, U_z\right) = s\left( \sqrt{U_x}, U_y, U_z\right) + O\left( \frac{\log \left(U_x U_y U_z\right)}{w} \right).
\]
Three rounds of recursion yield
\[s\left( U,U,U \right) = s\left( \sqrt{U},\sqrt{U},\sqrt{U} \right) +  O\left( \frac{\log U}{w} \right), \]
which solves to $s\left( U,U,U \right) = O\left( \frac{\log U}{w} \right)$.  As $U=2n$ initially, the total space in words is
$O\left(n\frac{\log n}{w}\right)\le O\left(n\right)$.
Note that the above analysis ignores an overhead of $O(1)$ words of space per node of the recursion tree, but by shortcutting degree-1 nodes, we can bound the number of nodes in the recursion tree by $O\left(n\right)$.
To summarize, we claim the following results:

\begin{theorem}\label{thm:opl3d}
Given $n$ disjoint axis-aligned boxes in 3-d,
there are data structures for point location with
$O\left(n\right)$ words of space and
 $O(\log n)$ query time in the pointer machine model,
$O(\log_B n)$ query cost in the I/O model, and 
$O(\log_w n)$ query time in the word RAM model.
\end{theorem}
\begin{proof}
The proof for the I/O model and the word RAM model can be found in  Section C of the appendix.
\end{proof}

Further applications of this framework to subdivisions, 4-d and higher dimensions are provided in Section D of the appendix.

\section{Rectangle Stabbing}
 \subsection{Preliminaries}

\begin{lemma}\label{lemma:slow}
(Rahul~\cite{r15}) There is a data structure of size $O(n)$ words which can answer a $5$-sided 3-d rectangle stabbing 
query in $O(\log^2 n\cdot\log\log n + k)$ time.
\end{lemma}

\begin{lemma}\label{lemma:leaf}
(Leaf structure.) For a set of size $O(w^{1/4})$, 
there is a data structure of size $O(w^{1/4})$ words which can answer
 a $5$-sided 3-d rectangle stabbing  query in $O(1 + k)$ time.
 \end{lemma}


\subsection{3-d $5$-sided rectangle stabbing}\label{sec:5-sided}

\begin{figure}[t]
 \centering
\includegraphics[scale=1]{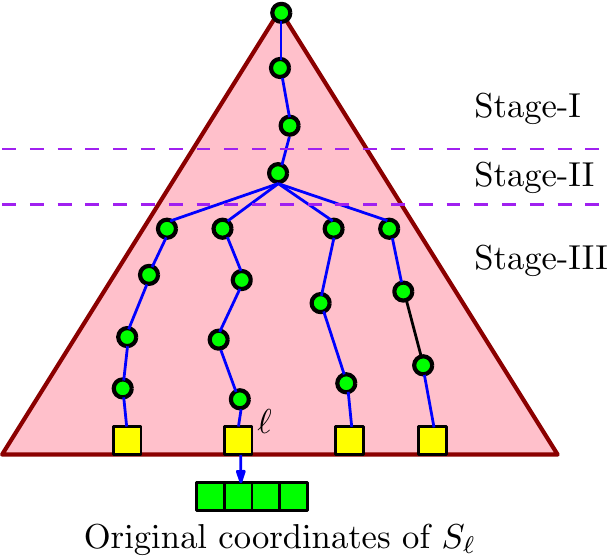}
\caption{Global Structure}
\label{fig:overall-structure}
\end{figure}


\noindent
{\bf Skeleton of the structure.} 
 Consider the projection of the rectangles of $S$ on to the $xy$-plane and impose  an orthogonal 
$\IntRange{2\sqrt{\frac{n}{\log^4 n}}} \times \IntRange{2\sqrt{\frac{n}{\log^4n}}}$ grid such that 
each horizontal and vertical slab contains the projections of $\sqrt{n\log^4n}$ sides of $S$.
This grid is the root node of our tree ${\cal T}$.
 For each vertical and horizontal slab, we recurse on the rectangles of $S$ which are {\em sent} to that slab. 
At each node of the recursion tree, if we have $m$ rectangles in the subproblem,
the grid size changes to  
$\IntRange{2\sqrt{\frac{m}{\log^4m}}} \times \IntRange{2\sqrt{\frac{m}{\log^4m}}}$.
We stop the recursion when a node has less than $w^{1/4}$ rectangles.

\subparagraph{Breaking the rectangles.}
The solution of Rahul~\cite{r15} {\em breaks} only one side to reduce  $5$-sided rectangles 
to $4$-sided rectangles, and then uses the solution for $4$-sided rectangle stabbing as a black box.
Unlike the approach of Rahul~\cite{r15}, we will break 
two sides of each $5$-sided rectangle to obtain $O(\log\log n)$ $3$-sided rectangles.

For a node in the tree, 
the intersection of every pair of  horizontal and 
vertical grid line defines a {\em grid point}. 
A rectangle $r\in S$ is associated with four root-to-leaf paths (as shown in
Figure~\ref{fig:overall-structure}). 
Any node (say, $v$) on  these four paths is classified w.r.t. $r$ into one of the three stages as follows:

\vgap
\noindent
{\it Stage-I.} The $xy$-projection of $r$ intersects none of the grid points.
Then $r$ is not stored at $v$, and sent to the child  
corresponding to the row or column  $r$ lies in.

\vgap
\noindent
{\it Stage-II.} The $xy$-projection of $r$  intersects at least one of the grid points. 
Then $r$ is broken into at most five disjoint pieces. 
 The first piece is a {\em grid rectangle}, which  is the bounding box of all the grid points lying inside $r$, as shown in 
Figure~\ref{fig:type-II}(b). The remaining four pieces are 
 two {\em column rectangles} and two {\em row rectangles} as shown in Figure~\ref{fig:type-II}(c)~and~(d), respectively.  
The grid rectangle is stored at $v$. Note that each column rectangle (resp., row rectangle) is now 
a $4$-sided rectangle in $\IR^3$ w.r.t. its column (resp., row), and is sent 
to its corresponding child node.

\begin{figure}[h]
 \centering
\includegraphics[scale=0.4]{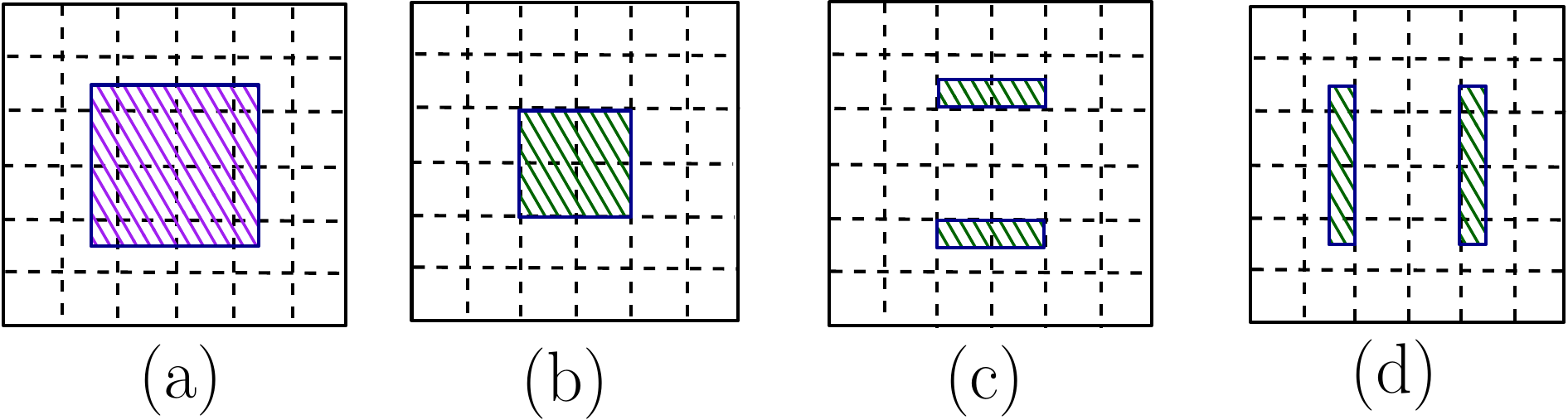}
\caption{Breaking a $5$-sided rectangle in (a) into 2 horizontal side rectangles (shown in (c)) and 2 vertical side rectangles (shown in (d)).}
\label{fig:type-II}
\end{figure}


\noindent
{\it Stage-III.} The $xy$-projection of a $4$-sided piece of $r$ intersects at least one of 
the grid points. Without loss of generality, assume that the $4$-sided rectangle $r$ is 
unbounded along the negative $x$-axis. 
Then the rectangle is broken into at most four disjoint pieces: 
 a {\em grid rectangle}, two {\em row rectangles}, and a  {\em column rectangle}, 
as shown in   Figure~\ref{fig:type-III}(b), (c) and (d), respectively. 
The grid rectangle and the two row rectangles are stored at $v$, 
and the column rectangle is sent to its corresponding child node.
Note that the two row rectangles are now $3$-sided rectangles in $\IR^3$ w.r.t. their 
corresponding rows (unbounded in one direction along $x$-, $y$- and $z$-axis).

\begin{figure}[h]
 \centering
\includegraphics[scale=0.4]{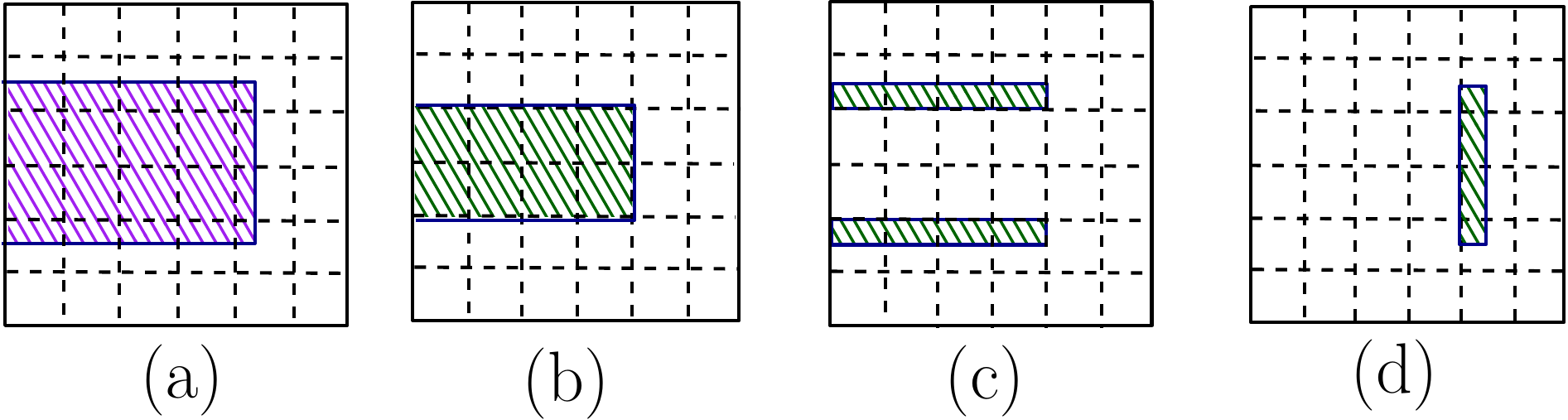}
\caption{Breaking a $4$-sided rectangle in (a) into 2 horizontal side rectangles (shown in (c)) and 2 vertical side rectangles (shown in (d)).}
\label{fig:type-III}
\end{figure}


\noindent
{\bf Encoding structures.}  Let $S_v$ be the set of rectangles stored at a 
node $v$ in the tree. We apply a rank space reduction (replacing input coordinates by their ranks) 
so that the  coordinates of all the endpoints are in $[2|S_v|]^3$.
If $v$ is a leaf node, then we build an instance of Lemma~\ref{lemma:leaf}.
Otherwise, the following three structures will be built using $S_v$:

\vgap
\noindent
{\it (A) Slow structure.} An instance of 
Lemma~\ref{lemma:slow} is built on $S_v$  to answer the 3-d $5$-sided rectangle stabbing 
query when the output size is ``large''. 

\noindent
{\it (B) Grid structure.}
For each cell $c$ of the grid, among the rectangles which completely cover $c$, 
pick the $\log^3|S_v|$ rectangles with the largest span along the $z$-direction. 
Store them in a list $\mathit{Top(c)}$ in decreasing order of their span.  

\noindent
{\it (C) 3-d dominance structure.} 
For a given row or column in the grid, based on the $3$-sided rectangles stored in it, 
a linear-space 3-d dominance reporting structure~\cite{p08b,c13} is built.
This structure is built for  each row and column slab.

\subparagraph{Where are the original coordinates stored?}
Unlike the previous approaches for indexing points~\cite{abr00,clp11}, 
we use a somewhat unusual approach for storing the original coordinates 
of each rectangle. 
In the process of breaking each $5$-sided rectangle described above, 
there will be four leaf nodes where portions of the rectangle 
will get stored. We will choose these leaf nodes to store the original coordinates of the rectangle 
(see Figure~\ref{fig:overall-structure}).
The benefit is that each $3$-sided rectangle (stored at a node $v$)
has to maintain a decoding pointer of length merely $O(\log |S_v|)$ to point to its original coordinates 
stored in its subtree.

\subparagraph{Query algorithm and analysis.} Given a query point $q$, we start at the root node and perform 
the following steps:
First, query the dominance structure corresponding to the horizontal and the vertical 
slab containing $q$.
Next, for the grid structure, locate the cell $c$ on the grid containing $q$. 
Scan the list $\mathit{Top(c)}$  to keep reporting  till (a) all the rectangles 
have been exhausted, or (b) a rectangle not containing $q$ is found. 
If case~(a) happens and $|\mathit{Top(c)}|=\log^3|S_v|$,  then we 
discard the rectangles reported till now, and query the slow structure. 
The decoding pointers will  be used to report the original coordinates 
of the rectangles.
Finally, we recurse on the horizontal and the vertical slab containing $q$. 
  If we visit a leaf node, then we query the leaf structure 
(Lemma~\ref{lemma:leaf}).

First, we analyze the space. 
Let $s(|S_v|)$ be the {\em amortized} number of bits needed per 
input $5$-sided rectangle in the subtree of a node $v$. The amortized number of bits needed 
per rectangle for the  encoding structures and the pointers 
to the original coordinates is $O(\log |S_v|)$. This leads to 
the following recurrence:
\[s(n)=s(\sqrt{n\log^4n}) + O(\log n) \]
which solves to $s(n)=O(\log n)$ bits. Therefore, the overall 
space is bounded by $O(n)$ words.

Next, we analyze the query time. To simplify the analysis, we will exclude the output size term while mentioning the query time. At the root, the time taken to query the grid and the dominance structure is $O(\log\log_wn)$. This leads to the following recurrence: 
\[Q(n)=2Q(\sqrt{n\log^4n}) + O(\log\log_wn) \]
with a base case of $Q(w^{1/4})=O(1)$. This solves to $Q(n)=O(\log_wn-\log\log_wn)=O(\log_wn)$.
For each reported rectangle it takes constant time to recover its original coordinates.
The time taken to query the slow structure is dominated by the output size. Therefore, the overall 
query time is $O(\log_wn + k)$.
\begin{theorem}\label{thm:rs-5-sided}

There is a data structure of size $O(n)$ words  
which can answer any 3-d $5$-sided rectangle stabbing query in $O(\log_wn + k)$ time. This is optimal in the 
word RAM model.
\end{theorem}

Our solution for 2-d top-$k$ rectangle stabbing can be found in Section E of the appendix.

\subsection{3-d $6$-sided rectangle stabbing}
In this section we will prove the following result. 

\begin{theorem}
\label{theor:stab6sid-main}
There is a linear-space data structure that answers 3-d rectangle stabbing queries in $O(\log^2_w n+k)$ time.
\end{theorem}
The complete discussion on $6$-sided rectangle stabbing 
is provided in Section F of the appendix. Here we will only highlight the key result.

\begin{lemma}
\label{lemma:rest4sid1-main}
  There exists a linear-space data structure that answers $z$-restricted 3-d $4$-sided rectangle 
  stabbing queries in $O(\log w\cdot \log\log n+k)$ time. A $z$-restricted $4$-sided rectangle 
  is of the form $(-\infty,x] \times (-\infty,y] \times [i,j]$,  where integers $i,j\in \left[ w^\eps \right]$ 
  and $\eps=0.1$.
\end{lemma}

\begin{lemma}
\label{lemma:restric4sid2-main}
  There exists an optimal linear-space data structure that answers $z$-restricted 3-d 4-sided rectangle stabbing
   queries in $O(\log\log_w n+k)$ time. A $z$-restricted $4$-sided rectangle 
  is of the form $(-\infty,x] \times (-\infty,y] \times [i,j]$,  where integers $i,j\in \left[ w^\eps \right]$ 
  and $\eps=0.1$.
\end{lemma}
\begin{proof}
 We can safely assume that $n > w^{2\eps}\cdot\log w\log\log n$, 
 because the case of $n< w^{2\eps}\cdot\log w\log\log n=O(w^{1/4})$ can be handled in 
 $O(1+k)$ time by using the structure of Lemma~\ref{lemma:leaf}.
 To keep the discussion short, we will assume that  $k<\log w\cdot\log\log n$ 
 (handling small values of $k$ is typically more challenging).

\subparagraph{Shallow cuttings.} A point $p_1$ is said to {\em dominate} 
point $p_2$ if it has a larger $x$-coordinate and a larger $y$-coordinate value. 
Our main tool to handle this case are {\em shallow cuttings} which have the following 
three properties: (a) A  $t$-shallow cutting for a set $P$ of 2-d  points 
is  a union of $O(n/t)$ cells where every cell is of the form $[a,+\infty)\times [b,+\infty)$, 
(b) every point that is dominated by at most $t$ points from $P$ will lie within some cell(s), and 
(c) each cell contains at most $O(t)$ points of $P$.
A cell $[a,+\infty)\times [b,+\infty)$ can be identified by its corner $(a,b)$. 
We denote by $\textit{Dom}(c)$ the set of points that dominate the corner $c$.

\subparagraph{Data structure.}  
We classify rectangles according to their $z$-projections. The set $S_{ij}$ contains 
  all rectangles of the form $r=(-\infty,x_f]\times (-\infty,y_f]\times [i,j]$. 
  Since $1\le i\le j\le w^{\eps}$, there are $O(w^{2\eps})$ sets $S_{ij}$.   
  Every rectangle $r$ in $S_{ij}$ is associated with a point $p(r)=(x_f,y_f)$. 
  We construct a $t$-shallow cutting $L_{ij}$ with $t=\log w\cdot\log\log n$ 
  for the set of points $p(r)$, such that $r\in S_{ij}$.  
  A rectangle $r=(-\infty, x_f]\times (-\infty, y_f]\times [i,j]$ 
  is stabbed by a query point $q=(q_x,q_y,q_z)$ if and only if  
  $p(r)\in S_{ij}$  and the point $p(r)$ dominates the 2-d point $(q_x,q_y)$. 
We can find points of a set $S_{ij}$ that dominate $q$ using the shallow cutting $L_{ij}$. 
However, to answer the stabbing query we must simultaneously answer   
a dominance query on $O(w^{2\eps})$ different sets of points.

We address this problem by {\em grouping} corners of different shallow cuttings into one structure. 
Let $\cC_{ij}$ denote the set of corners in a shallow cutting $L_{ij}$ and let 
$\cC=\bigcup_{\forall i,j \in [w^{\eps}]} \cC_{ij}$. The set $\cC$ is divided into disjoint groups, 
so that every group $G_{\alpha}$ consists of $w^{2\eps}$ consecutive corners 
(with respect to their $x$-coordinates): for any $c\in G_{\alpha}$ and $c'\in G_{\alpha+1}$, 
$c.x< c'.x$. We say that a corner $c\in \cC_{ij}$ is immediately to the left of 
$G_{\alpha}$ if it is the rightmost corner in $\cC_{ij}$ such that $c_x\le c'_x$ for any corner $c'=(c'_x,c'_y)$ in $G_{\alpha}$.  The set of corners $\oG_{\alpha}$ contains (1) all corners from $G_{\alpha}$, and  
(2) for every pair $i,j$ such that $1\le i\le j\le w^{\eps}$, the corner $c\in \cC_{ij}$ immediately to the left of $G_{\alpha}$. The set $R_{\alpha}$ contains all rectangles $r$ such that $p(r)\in \textit{Dom}(c)$ for each corner 
$c\in \oG_{\alpha}$.  Since $R_{\alpha}$ contains $O(w^{2\eps}\log w\cdot\log\log n)=O(w^{1/6})$ rectangles, we can perform a rank-space reduction and answer queries on $R_{\alpha}$ in $O(k+1)$ time by using Lemma~\ref{lemma:leaf}; see Fig.~\ref{fig:corners}. 

\begin{figure}[h]
\includegraphics[width=.5\textwidth]{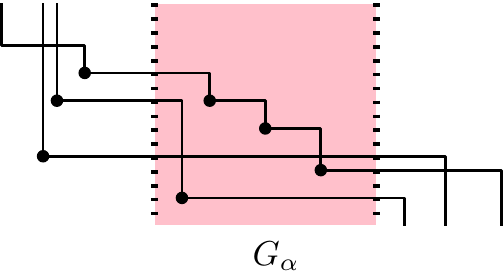}
\caption{The set of corners $G_{\alpha}$}
\label{fig:corners}
\end{figure}

Next, we will show that the space occupied by this structure is $O(n)$. The crucial observation is that 
the number of corners in $G_{\alpha}$ is $w^{2\eps}$ and 
the number of ``immediately left'' corners added to each $\oG_{\alpha}$ is also bounded by $w^{2\eps}$. 
The number of corners in set  $\cC$ is bounded by 
$\sum_{\forall L_{i,j}}O\left(1+\frac{|S_{ij}|}{t}\right)=O(n/t)$, 
since $n/t > w^{2\eps}$. Therefore, the number of groups will be $O\left(\frac{n}{tw^{2\eps}}\right)$. 
Each set $R_{\alpha}$ contains $O(w^{2\eps}t)$ rectangles. Therefore, the total space occupied by this 
structure is $\sum_{\forall \alpha} |R_{\alpha}|=O(\frac{n}{tw^{2\eps}}\cdot w^{2\eps}t)=O(n)$.

\subparagraph{Query algorithm.} 
Given a query point $q=(q_x,q_y,q_z)$, we find the set $G_{\alpha}$ 
that ``contains'' $q_x$. Then we report all the rectangles in $R_{\alpha}$ that are stabbed by $q$ 
by using Lemma~\ref{lemma:leaf}. We need $O(\log\log_w n)$ time to find the group $G_{\alpha}$~\cite{pt06} and, then $O(1+k)$ time to  report $R_{\alpha} \cap q$. 

Our procedure reports all rectangles stabbed by the query point: Suppose that a point $q$ stabs at most $t$ rectangles. Let $G_{\alpha}$ denote the group of corners that contains $q$. If $q$ stabs a rectangle $r$, then $p(r)$ dominates $q$. Hence $p(r)\in \textit{Dom}(c)$ for some corner $c\in \oG_{\alpha}$ and the rectangle $r$ is stored in the data structure $R_{\alpha}$. Now suppose that  $q$ stabs more than $t$ rectangles. Then, by the same token, there are at least 
$t$ rectangles in $R_{\alpha}$ that are stabbed by $q$. Hence we will use the slow data structure to answer the query and correctly report all rectangles in this case. 
\end{proof}


\bibliography{ref-pl}

\newpage


\section*{\Huge{Appendix}}

\appendix
\section{On the Models}\label{app:models}
Throughout this paper, the \emph{pointer machine model} refers to
the ``arithmetic pointer machine'' (APM) in the terminology
from Chazelle's paper~\cite{c88}: Each word (or memory cell) stores
a constant number of pointers, input points, and/or $w$-bit integers for a fixed $w$.  We support pointer chasing and standard arithmetic operations, comparisons, and shifts on $w$-bit integers in unit time each, but do not allow pointer arithmetic.  It is assumed that $w\ge\log n$ (which is reasonable since a pointer or input point requires $\Omega(\log n)$ bits).

In the \emph{I/O model}, each block is assumed to hold $B$ words for a fixed $B$,
where each word stores an input point or a $w$-bit integer, again assuming that $w\ge\log n$.  We support block reads/writes with unit cost each; all other operations on a block are free.

In the \emph{word RAM model}, each word stores a $w$-bit integer, again assuming that $w\ge \log n$; we support standard arithmetic operations, comparisons, bitwise logical operations, and shifts on $w$-bit integers in unit time each, and allow these $w$-bit integers to be used as pointers.  Furthermore, it is assumed that the coordinates of the input points are $w$-bit integers.

\section{Proof of Lemmata~\ref{lem:opl2d} and~\ref{lem:stab2d}}\label{app:missing-details}
For Lemma~\ref{lem:opl2d}, such data structures for
2-d orthogonal point location can be found in \cite{lt80,k83,egs86,st86,snoeyink} for the pointer machine model,  \cite{gtvv93,adt03} for the I/O model, and \cite{c13} for the word RAM model.  For Lemma~\ref{lem:stab2d},  2-d rectangle stabbing emptiness (or more generally, rectangle stabbing counting) is known to be reducible to 
2-d orthogonal range counting~\cite{eo82}, and
such data structures for 2-d orthogonal range counting can be found in~\cite{c88} for the pointer machine model, \cite{gaa03} for the I/O model, and \cite{jms04} for the word RAM model.

All these known data structures technically require $O(n)$ words of space, or more precisely, $O(n\log U)$ bits of space.  In the I/O model or word RAM model, we can easily pack the data structures in $O\left(\frac{n\log U}{w}\right)$ words of space without increasing the query cost when $\log U\ll w$.  In the pointer machine model, we may not be able to pack the data structures in general, since
if multiple ``micro-pointers'' are packed in a word, the model does not allow us to follow such a micro-pointer.  Nevertheless, it is not difficult to modify the existing data structures to achieve the compressed space bound; next, we present the technical details of the modifications needed.

\subparagraph*{Proof of Lemma~\ref{lem:opl2d} for pointer machines.}

For 2-d orthogonal point location, one solution is via \emph{$(1/r)$-cuttings}~\cite{gor97}:
we can partition the plane into $O(r)$ disjoint rectangular cells, each intersecting $O(n/r)$ line segments (edges of the input rectangles), where 
we choose $r=\frac{\delta n\log U}{w}$ for a sufficiently small constant $\delta$. 

We build a point location structure~\cite{lt80,k83,egs86,st86,snoeyink} for the $O(r)$ cells with  $O(\log r)$ query time in the pointer machine model; the space usage of this structure in words is $O(r)$, which is within the allowed bound $O\left(\frac{n\log U}{w}\right)$, so there is no need for bit packing here.  

For each cell, we store the $O(n/r)$ line segments in a point location structure~\cite{k83}  with 
$O\left(\log(n/r)\right)$ query time; the space usage of this structure in bits is
$O\left((n/r)\log U\right)$, which is $O\left(\delta w\right)$, so the entire structure can be packed in a single word.
Although pointer chasing is not directly supported in the pointer machine model when multiple ``micro-pointers'' are packed in a word, we can simulate each pointer chasing step here in constant time by arithmetic operations and shifts within the word.

Given a query point $q$, we can first find the cell containing $q$ in $O\left(\log r\right)$ time and then finish the query inside the cell in $O\left(\log(n/r)\right)$ time.
The overall query time is $O\left(\log n\right)$.

\subparagraph*{Proof of Lemma~\ref{lem:stab2d} for pointer machines.}

Rectangle stabbing emptiness in 2-d reduces to dominance range counting in 2-d~\cite{eo82}.
Chazelle's \emph{compressed range tree} structure~\cite{c88} solves the latter problem with $O(n)$ words of space and $O(\log n)$ time in the pointer machine model.  We observe that his data structure actually
achieves $O\left(\frac{n\log U}{w}\right)$ words of space, after minor modifications.

At each level of the range tree, Chazelle's structure stores lists consisting of a total of $O\left(\frac{n}{w}\right)$ words ($O\left(\frac{n}{w}\right)$ $w$-bit integers as well as $O\left(\frac{n}{w}\right)$ pointers to words in lists at the next level).  The total number of words over all levels of the tree is $O\left(\frac{n\log n}{w}\right)\le O\left(\frac{n\log U}{w}\right)$.

We shorten the tree by making the leaf nodes contain $b$ points, where we choose $b = \frac{\delta w}{\log U}$ for a sufficiently small constant $\delta$.  This way, the space in words for the tree itself is
$O\left(n/b\right) = O\left(\frac{n\log U}{w}\right)$.
Inside each leaf, we store the $b$ points in another instance of
Chazelle's structure; the space usage of this structure in bits is
$O\left(b\log U\right)$, which is $O\left(\delta w\right)$, so the entire structure can be packed in a single word.
Again, we can simulate each pointer chasing step here in constant time by arithmetic operations and shifts within the word.

To answer a dominance range counting query, we descend along a path
in the compressed range tree, which requires $O(\log (n/b))$ time by following pointers in the lists stored at the path and doing various arithmetic operations and shifts on $w$-bit integers.  At the leaf of the path, we can finish the query in $O\left(\log b\right)$ time.
The overall query time is $O\left(\log n\right)$.

\section{Other Models}\label{app:other-models}

In the I/O model, the analysis is similar, with a modified recurrence for the query cost:
\[
Q\left( U, U, U\right) = Q\left( \sqrt{U}, \sqrt{U}, \sqrt{U}\right) + O\left( \log_B U \right).
\]
For the base case $U\le B^{1/3}$, we have $Q\left( U, U, U\right)=O(1)$ trivially, since $n\le U^3\le B$.
Solving the recurrence yields $O\left( \log_B N\right)$ query cost.
The space usage remains $O\left(N\right)$ words (i.e., $O\left( N/B \right)$ blocks).

In the word RAM model, the analysis is again similar, with
\[
Q\left( U, U, U\right) = Q\left( \sqrt{U}, \sqrt{U}, \sqrt{U}\right) + O\left( \log_w U\right).
\]
For the base case $U\le w$, we have $Q\left( U,U,U\right) = O\left(1\right)$ by switching to another known method:
Orthogonal point location in 3-d reduces to 6-d dominance emptiness, for which there is a known
method~\cite{cz15} with $O\left( n(\log_w n)^4\right)$ words of space
and $O\left( (\log_w n)^5\right)$ query time in the word RAM\@.
(The method in~\cite{cz15} can be modified to report a witness if the
range is non-empty.)
Since $n\le U^3\le w^3$, we have $\log_w n = O(1)$, and so the space bound is $O\left( n\right)$ and
query bound is $O\left( 1\right)$ for the base case.
Solving the recurrence yields $O\left( \log_w N\right)$ query time.

\section{Final Remarks}\label{app:higher-subdvision}

\subparagraph{Orthogonal point location in 4-d.} 
In the word RAM model, if we plug in the bounds  for 3-d 6-sided rectangle stabbing 
(Theorem~\ref{theor:stab6sid}) and the bounds for 3-d orthogonal point location (Theorem~\ref{thm:opl3d}) 
into our framework, then we obtain a linear-space structure which can 
answer a  4-d orthogonal point location query in $O(\log_w^2n)$ time.

\subparagraph*{Higher dimensions.}
The same approach can be extended to higher dimensions, reducing the
complexity of $d$-dimensional orthogonal point location to that
of $(d-1)$-dimensional box stabbing emptiness.  However, known data structures for higher-dimensional box stabbing~\cite{aal12} requires superlinear
space, whereas
the simpler approach mentioned in the Introduction, 
of using interval trees to reduce the dimension, gives $O\left(\log^{d-2}n\right)$
query time while keeping linear space in the pointer machine model.

\subparagraph*{The case of 3-d subdivisions.}
Our approach can also be used to improve the space bound of
de Berg, van Kreveld, and Snoeyink's point location structure~\cite{bks95} for 3-d orthogonal subdivisions, from $O\left(n \log\log U\right)$ space to $O\left( n\right)$, in the word RAM model.

\begin{theorem}\label{thm:sub}
Given a subdivision formed by $n$ disjoint (space-filling) axis-aligned boxes in 3-d,
there is a data structure for point location with
$O\left(n\right)$ words of space and $O\left( \log^2 \log U \right)$ query time in the word RAM model.
\end{theorem}
\begin{proof} (Sketch)\ \
De Berg et al.'s method~\cite[Theorem 2.4]{bks95} was already based on a van Emde Boas recursion, partitioning along the $x$-direction.  They also used 2-d orthogonal point location structures during the recursion, but managed to avoid rectangle stabbing structures by exploiting the fact that the input is a
subdivision.
Roughly, for each slab, they took the ``holes'' formed by all middle boxes that intersect the slab, and filled the holes by taking the vertical decomposition of the $yz$-projection.  The analysis followed by charging the complexity of the decomposition to vertices within the slab.

Our new change is to do the van Emde Boas recursion not just along the $x$-direction but along all three axis directions in a round-robin fashion.  This leads to the same recurrence for space as in Section~\ref{sec:opl3d}.  The query time satisfies the following recurrence:
\[
Q\left( U, U, U\right) = Q\left( \sqrt{U}, \sqrt{U}, \sqrt{U}\right) + O\left( \log\log U\right).
\]
This leads to $O(\log^2\log U)$ query time.
\end{proof}

\section{Top-$k$ 2-d rectangle stabbing}\label{app:top-k}
Our solution for $5$-sided rectangle stabbing can be modified to support top-$k$ stabbing queries in optimal time. We use the same general approach, but we need additional ideas to handle the {\em top-$k$}
aspect of the problem.

\subsection{Top-$k$ 2-d dominance query}

First, we will present an optimal solution for
the top-$k$ 2-d dominance query where the input is a set $P$ of $n$ weighted points in 2-d, and the query is
an integer $k$ and a dominance range $q=[q_x,\infty) \times [q_y,\infty)$. The result obtained
is the following.
\begin{theorem}
  \label{theor:topkdomin}
There exists a linear-space data structure that answers the top-$k$ 2-d dominance query in $O(\log\log_w n + k)$ time. This is optimal in the word RAM model.
\end{theorem}

\subparagraph{Preliminaries.} The optimality of Theorem~\ref{theor:topkdomin} follows from the lower bound of Patrascu and Thorup for the predecessor search problem~\cite{pt06}. We will need the following two building blocks for our solution.
\begin{lemma} (Patil {\em et al.}\cite{pts+14}, Theorem 9)
  \label{lemma:topkdominslow}
There exists a linear-space data structure that  answers the top-$k$ 2-d dominance query in $O(\log n + k)$ time. The points are reported in a sorted order.
\end{lemma}

\begin{lemma}
  \label{lemma:topksmalldomin}
We can keep $m=\log^{1/3}n$ points in a data structure of size $O(m)$-words   that
 answers the top-$k$ 2-d dominance query
 in $O(1+k)$ time.
\end{lemma}
\begin{proof}
  First, we reduce the problem to rank space.
  Next, we build an instance of Lemma~\ref{lemma:topkdominslow} on the rank-reduced dataset.
  Finally,  since there are only $O(m^2)$ combinatorially different queries we  store $\log\log n$
   highest weighted points that dominate each query point. As explained in the proof of Lemma~\ref{lemma:leaf} we can keep all pre-computed solutions in $O(m)$ words.

   To answer a query, the case of $k<\log\log n$ is handled by using the pre-computed solution, and
   the case of  $k> \log\log n$ is handled in $O(\log\log n + k)=O(k)$ time by querying the structure of Lemma~\ref{lemma:topkdominslow}.
\end{proof}

\subparagraph{Shallow cuttings.}
We will re-define shallow cuttings in the context of 3-d points.
A point $p_1$ is said to {\em dominate}
point $p_2$ if it has a larger coordinate value in all the three dimensions.
A $t$-shallow cutting for a set $P$ of 3-d points  is a collection of \emph{boxes}
of the form $[a,+\infty)\times [b,+\infty)\times [c,+\infty)$,
such that (a) there are only $O(n/t)$ boxes,
(b) every point that is dominated by at most $t$ points from $P$ will lie within some box,
and (c) each box contains at most $O(t)$ points of $P$.
A box $[a,+\infty)\times [b,+\infty)\times [c,\infty)$ can be identified by its corner $(a,b,c)$.

A common operation on a shallow cutting is
FIND-ANY: Given a 3-d point $q'$, find any box in the shallow cutting which contains $q'$.
The standard implementation of this query leads to a planar subdivision in the x-y plane
consisting of orthogonal rectangles where each rectangle is  labeled by a box in the
shallow cutting. Now given a  point $q'(q_x,q_y,q_z)$, we first perform a point location
query on the planar subdivision with $(q_x,q_y)$.

If
$B$ is the label on the rectangle and if $B$ contains $q$, then we report $B$;
otherwise, we can safely conclude that no box contains $q$.

\subparagraph{Data structure.}
Our structure consists of the following components:\\
(A) {\em Slow structure.} Based on the pointset $P$, we build the data structure of Lemma~\ref{lemma:topkdominslow}.\\
(B) {\em $\log n$-level shallow cutting.} We regard the weights of the points of $P$ as the third coordinate and  construct a $\log n$-shallow cutting $\cP_1$. \\
(C) {\em Slow structure for each box.} The {\em conflict list}, $CL_B$, of a box $B\in \cP_1$ is the points in  $P$ which lie inside it. For each $CL_B$ the data structure of Lemma~\ref{lemma:topkdominslow} is constructed.\\
(D) {\em $\log^{1/3}n$-level shallow cutting.} A $\log^{1/3}n$-shallow cutting $\cP_2$ is constructed based on points in $P$.\\
(E) {\em Small-sized structures.} For every box in $\cP_2$, based on its conflict list
 build the structure of   Lemma~\ref{lemma:topksmalldomin}.

The properties of the shallow cuttings and the fact that Lemma~\ref{lemma:topkdominslow} and
Lemma~\ref{lemma:topksmalldomin} are linear-space structures ensures that the space occupied by our
data structure is $O(n)$ words.

\subparagraph{Query algorithm.} We divide the query algorithm into three cases:\\
(A) $k\geq \log n$: Then we query the slow structure.\\
(B) $k\in [\log^{1/3}n,\log n]$: Then we perform the FIND-ANY operation on the $\log n$-shallow cutting $\cP_1$ to find a box $B$ whose projection contains $q$, and then query the slow structure built on $CL_B$.\\
(C) $k\in [1,\log^{1/3}n]$: Then we perform the FIND-ANY operation on the $\log^{1/3} n$-shallow cutting $\cP_2$ to find a box $B$ whose projection contains $q$, perform a rank-space reduction of the query w.r.t. to $CL_B$, and then query the small-sized structure built on $CL_B$.

We need $O(\log\log_w n)$ time to answer the FIND-ANY operation~\cite{c13}.
 All the other steps take  $O(k)$ time. Thus the total query time is $O(\log\log_w n +k)$.

\subparagraph{Proof of correctness.} Assume that  $k\in [\log^{1/3}n,\log n]$.
Given a query point $q(q_x,q_y)$, let $q'=(q_x,q_y,q_z)$ be a point such that there are
exactly $k$ points in $P \cap q'$. Then it is guaranteed that there exists a box  in the $\log n$-shallow cutting which will contain $q'$. The implementation of the FIND-ANY operation ensures that the box $B$ returned by it will contain $q'$. Therefore, the top-$k$ points of $P \cap q$ will be present in the
conflict list of $B$. A similar argument holds for $k<\log^{1/3}n$.

\subparagraph{Remark.} In the query algorithm, we assumed that $|S\cap q| \geq k$.
This is easy to check by reporting the points in $S\cap q$ till one of the
following happens: either $k+1$
points are reported or all the points are reported.

\subsection{Top-$k$ 2-d rectangle stabbing}
Now we are ready to prove the following result.
\begin{theorem}
\label{theor:topkstab}
There is a linear-space data structure  that  can answer  any top-$k$ 2-d rectangle stabbing query in $O(\log_wn+ k)$ time. This is optimal in the word RAM model.
\end{theorem}

\subparagraph{Preliminaries.} As in the case of top-$k$ 2-d dominance query, we will need a
slow structure and a small-sized structure.

\begin{lemma}
\label{lemma:topkstabslow}
  There is a linear-space data structure which can answer any  top-$k$
  2-d rectangle stabbing query in $O(\log^2 n\cdot\log\log_w n +k)$ time.
\end{lemma}
\begin{proof}
  It is known that we can answer  a 3-d 5-sided  rectangle stabbing  query by querying
  $O(\log^2n)$ 3-d dominance reporting structures~\cite{r15}. The space of the data structure
  is within the same  bounds as the 3-d dominance reporting structure.  Refer to \cite{r15}
  for further details.

In the same way, we can answer a 2-d top-$k$  rectangle stabbing query by querying $\log^2n$  top-$k$ 2-d dominance queries. We will return the $k$ heaviest points in {\em sorted} order using the same technique that will be used later in this section.
\end{proof}

\begin{lemma}
  \label{lemma:topk-stab-small}
We can keep $m=w^{1/4}$ points in a data structure of size $O(m)$-words   that
 answers the top-$k$ 2-d rectangle stabbing query
 in $O(1+k)$ time.
\end{lemma}

All the above data structures (Theorem~\ref{theor:topkdomin}, Lemma~\ref{lemma:topkstabslow}, and
Lemma~\ref{lemma:topk-stab-small}) support queries in  an online manner. That is, we do not need to know the value of $k$ when the query is asked, since the elements are reported  in descending order of their weights. The procedure can be paused and resumed at a later time.

\subparagraph{Data structure.} We only worry about $k<\log^3n$, since the other case can be
optimally handled by Lemma~\ref{lemma:topkstabslow}.
We use the same structure as in the proof of Theorem~\ref{thm:rs-5-sided} with the following differences: we use  top-$k$ 2-d dominance  structure of Theorem~\ref{theor:topkdomin} instead of the 3-d dominance structure,
 and we use the slow data structure of Lemma~\ref{lemma:topkstabslow} instead of the slow structure of
Lemma~\ref{lemma:slow}. Now $Top(c)$ is the $\log^3|S_v|$ heaviest rectangles among the rectangles covering the grid cell $c$ and is stored in decreasing order of their weight. Finally, the leaf structure is built by using Lemma~\ref{lemma:topk-stab-small}.

\subparagraph{Query algorithm.} Consider the following abstract problem:
We are given sorted lists $L_1,L_2,\ldots, L_t$ such that the total number of
elements in all the lists is less than or equal to $n$, and $t=O(\log_wn)$.
The goal is to report the $k$ heaviest elements. To answer this, we build a
heap $\cG$ based on the heaviest element from each list. Now we perform
the following operations on $\cG$ $k$ times: (a) delete the element, $e$, with the largest
weight and report it, and (b) if $e$ came from list $L_i$, then
insert the next heaviest element from $L_i$ into $\cG$.
 We implement $\cG$ as a fusion tree~\cite{fw93}; since $\cG$ contains at most $\log_w n$ elements, all operations on $\cG$ are supported in $O(1)$ time.

By now the reader must have guessed the query algorithm.\\
(A) {\em Identify} the $O(\log_w n)$ nodes as explained in Section~\ref{sec:5-sided}. \\
(B) Each  identified node acts a list $L_i$ (as discussed before, the grid structure and
the top-$k$ 2-d dominance structure can report in an online manner. If
more than $\log^3|S_v|$ rectangles have been reported from a node $v$, then
we switch to its slow structure).\\
(C) Now find the top-$k$ heaviest rectangles in $S\cap q$.

\section{3-d 6-sided rectangle stabbing}
\label{sec:6sid}
In this section we will fill the missing details for the proof 
of the following result. 

\begin{theorem}
\label{theor:stab6sid}
There is a linear-space data structure that answers 3-d rectangle stabbing queries in $O(\log^2_w n+k)$ time.
\end{theorem}

\subsection{Skeleton structure}

\subparagraph{Data structure.} We construct an interval tree, $IT$, with fan-out $w^{\eps}$ on the $z$-projections of rectangles in $S$ 
for a positive constant $\eps<1/18$. For a node $v \in IT$, let $\ell(v_i)$ denote the bounding planes 
of its children and let $z_i(v)$ be the $z$-coordinate of $\ell(v_i)$. 
The set $S(v) \subseteq S$ contains all rectangles $[x_1,x_2]\times [y_1,y_2]\times [z_1,z_2]$, 
such that $v$ is the lowest common ancestor of the leaves storing $z_1$ and $z_2$. We keep three stabbing data structures $M(v)$, $L(v)$, and $R(v)$ at each node $v$. Consider an arbitrary rectangle 
$r=[x_1,x_2]\times [y_1,y_2]\times [z_1,z_2]$ stored in $S(v)$. 
Suppose that $z_k(v)< z_1\le z_{k+1}(v)$ and $z_{l-1}(v)\le z_2< z_{l}(v)$. 
If $k+1<l-1$, we store a rectangle $[x_1,x_2]\times [y_1,y_2]\times [k+1,l-1]$ 
in a data structure $M(v)$. We also store $r$ in data structures $R(v_k)$ and $L(v_l)$. 
The $z$-coordinates of all rectangles in $M(v)$ lie in the integer universe $[w^{\eps}]$. 
We will  use this fact to answer rectangle stabbing queries in $O(\log_w n+k)$ time, as will be shown later in this section in Theorem~\ref{thm:z-restricted}. Rectangles stored in $L(v)$ and $R(v)$ cross the left or the right bounding plane of the node $v$. Hence we can treat the rectangles in $L(v)$ (resp.\ $R(v)$) as 5-sided rectangles. Using Theorem~\ref{thm:rs-5-sided}, we can answer stabbing queries on $R(v)$ and $L(v)$ in $O(\log_w n +k )$ time.

\subparagraph{Query algorithm.} To report all rectangles that stab a point $q$, we traverse a root-to-leaf path 
in $IT$ and answer stabbing queries using data structures built for $L(v)$, $R(v)$ and $M(v)$ in each node $v$. Since the length of a root-to-leaf path is  $O(\log_w n)$, the query is answered in $O(\log^2_w n+k)$ time

\subsection{$z$-restricted queries}  

It remains to show how to answer 3-d rectangle stabbing queries when the $z$-coordinates of the endpoints 
are bounded by the integer universe $[w^{\eps}]$.  This scenario will be called {\em $z$-restricted queries}.
\begin{lemma}
\label{lemma:rest4sid1}
  There exists a linear-space data structure that answers $z$-restricted $4$-sided rectangle 
  stabbing queries in $O(\log w\cdot \log\log n+k)$ time.
\end{lemma}
\begin{proof}
  We construct a (binary) interval tree $\cT_z$ on the $z$-projections of the rectangles. Every leaf of this tree corresponds to a $z$-coordinate in the universe $[w^{\eps}]$. Since $\cT_z$ has $w^{\eps}$ leaves, 
  its height is $\log w^{\eps}=O(\log w)$. Each rectangle is stored at a particular node in the tree and 
  for each node we build two data structures that support 3-d dominance reporting queries. See e.g., Rahul~\cite{r15} for a detailed description about the construction of the data structure. A stabbing query can be answered by traversing a path from the root of $\cT$ to a leaf node. 
 Since a 3-d dominance reporting query can be answered in $O(\log\log n + k)$ time and we visit $O(\log w)$ nodes, the total time needed to answer a query is $O(\log w\cdot \log\log n+k)$. 
\end{proof}

\begin{lemma}
\label{lemma:restric4sid2}
  There exists a linear-space data structure that answers $z$-restricted 4-sided 3-d rectangle stabbing
   queries in $O(\log\log_w n+k)$ time.
\end{lemma}
\begin{proof}
Lemma~\ref{lemma:restric4sid2-main} (in the main body of the paper) proves this lemma for the 
case $k<\log w\cdot\log\log_wn$. To handle the case of $k\geq \log w\log\log_wn$, 
we will use Lemma~\ref{lemma:rest4sid1}.

\end{proof}

Now we turn our attention to $z$-restricted $6$-sided rectangle stabbing queries. 
In this case, the data structure contains $6$-sided rectangles, 
but again the $z$-coordinates of the endpoints lie in the integer universe $[w^{\eps}]$.
\begin{lemma}
  \label{lemma:restrslow}
  There exists a linear-space data structure that answers $z$-restricted six-sided rectangle stabbing queries in $O(\log n + k)$ time. 
\end{lemma}
\begin{proof}
  Again we use a binary interval tree on the $z$-projections of the rectangles, then assign 
  each rectangle to a particular node in the tree, and then build 
  two data structures at each node which will answer 5-sided rectangle stabbing queries. When we answer a $6$-sided query, we traverse a root-to-leaf path and answer a 5-sided query at every node. We refer the reader to~\cite{r15} or \cite{ptsnv14} for a complete description.  The height of the tree is $O(\log w)$ and, by Theorem~\ref{thm:rs-5-sided} we need $O(\log_w n)$ time to answer a 5-sided query. Hence the overall time it takes to answer a $z$-restricted six-sided query is $O(\log n)$. 
\end{proof}

\begin{theorem} \label{thm:z-restricted}
  There exists a linear-space data structure that answers a 
  $z$-restricted 6-sided rectangle stabbing query in $O(\log_w n+ k)$ time. 
\end{theorem}
\begin{proof}
  We use the approach as in Theorem~\ref{thm:rs-5-sided}, but use different encoding structures.\\
  (A) {\it Slow structure.}  An instance of 
Lemma~\ref{lemma:restrslow} is built on $S_v$. \\
(B) {\it Grid structure.}
 For every cell $c$ we keep lists $\mathit{Cover(c,z)}$ for $z=1$, $\ldots$, $w^{\eps}$; 
 $\mathit{Cover(c,z)}$ contains $\log |S_v|$ rectangles $r$, such that the projection of $r$ onto the $xy$-plane completely covers $c$ and the $z$-projection of $r$ is stabbed by $z$.  \\
(C) {\it ``$z$-restricted dominance'' structure.} For a given row or column in the grid, 
based on the $z$-restricted 4-sided rectangles stored in it, 
 an instance of Lemma~\ref{lemma:restric4sid2} is built.
 
 The space and the query time analysis follow from the analysis of Theorem~\ref{thm:rs-5-sided}.
\end{proof}



\end{document}